\documentclass[11pt]{article}

\usepackage[letterpaper,margin=1.00in]{geometry}
\usepackage{amsmath, amssymb, amsthm, thmtools, amsfonts}
\usepackage{bbm}

\usepackage{ifthen}

\usepackage{tikz}
\usetikzlibrary{positioning,decorations.pathreplacing}

\usepackage{cite}
\usepackage{appendix}
\usepackage{graphicx}
\usepackage{color}
\usepackage{algorithm}
\usepackage[noend]{algpseudocode}
\usepackage{epstopdf}
\usepackage[textsize=tiny]{todonotes}

\usepackage{framed}
\usepackage[framemethod=tikz]{mdframed}
\usepackage[bottom]{footmisc}
\usepackage[shortlabels]{enumitem}
\setitemize{noitemsep,topsep=3pt,parsep=3pt,partopsep=3pt}
\usepackage[font=small]{caption}
\usepackage{xspace}

\newtheorem{theorem}{Theorem}[section]
\newtheorem{lemma}[theorem]{Lemma}
\newtheorem{meta-theorem}[theorem]{Meta-Theorem}

\definecolor{darkgreen}{rgb}{0,0.5,0}
\usepackage{hyperref}
\hypersetup{
    unicode=false,          
    colorlinks=true,        
    linkcolor=red,          
    citecolor=darkgreen,        
    filecolor=magenta,      
    urlcolor=cyan           
}
\usepackage[capitalize, nameinlink]{cleveref}
\crefname{theorem}{Theorem}{Theorems}
\Crefname{lemma}{Lemma}{Lemmas}
\Crefname{conjecture}{Conjecture}{Conjectures}

\algnewcommand\algorithmicswitch{\textbf{switch}}
\algnewcommand\algorithmiccase{\textbf{case}}

\algdef{SE}[SWITCH]{Switch}{EndSwitch}[1]{\algorithmicswitch\ #1\ \algorithmicdo}{\algorithmicend\ \algorithmicswitch}%
\algdef{SE}[CASE]{Case}{EndCase}[1]{\algorithmiccase\ #1}{\algorithmicend\ \algorithmiccase}%
\algtext*{EndSwitch}%
\algtext*{EndCase}%

\newcommand{\eps}{\varepsilon}

\newcommand{\congest}{$\mathsf{CONGEST}$\xspace}
\newcommand{\local}{$\mathsf{LOCAL}$\xspace}

\newcommand{\poly}{\operatorname{\text{{\rm poly}}}}

\renewcommand{\paragraph}[1]{\vspace{0.15cm}\noindent {\bf #1}:}

\newcommand{\FullOrShort}{full}

\ifthenelse{\equal{\FullOrShort}{full}}{
	
  \newcommand{\fullOnly}[1]{#1}
  \newcommand{\shortOnly}[1]{}
  
  }{

    \newcommand{\fullOnly}[1]{}
    \newcommand{\IncludePictures}[1]{}
   
  }

\begin{document}

\date{}

\title{Distributed Coloring for Everywhere Sparse Graphs}

\author{
 Mohsen Ghaffari\\
  \small ETH Zurich \\
  \small ghaffari@inf.ethz.ch
\and
	 Christiana Lymouri\\
  \small ETH Zurich \\
  \small lymouric@student.ethz.ch
	\and
 }

\maketitle

\setcounter{page}{0}
\thispagestyle{empty}
\begin{abstract}

Graph coloring is one of the central problems in \emph{distributed graph algorithms}.
Much of the research on this topic has focused on coloring with $\Delta+1$ colors, where $\Delta$ denotes the maximum degree.
Using $\Delta+1$ colors may be unsatisfactory in sparse graphs, where not all nodes have such a high degree;
it would be more desirable to use a number of colors that improves with sparsity.
A standard measure that captures sparsity is \emph{arboricity}, which is the smallest number of forests into which the edges of the graph can be partitioned.

\medskip
We present simple randomized distributed algorithms that,
with high probability, color any $n$-node $\alpha$-arboricity graph:
\begin{itemize}[leftmargin=1.2cm]
\item using $(2+\eps)\cdot \alpha$ colors, for constant $\eps>0$, in $O(\log n)$ rounds, if $\alpha=\tilde{\Omega}(\log n)$, or
\item using $O(\alpha \log \alpha )$ colors, in $O(\log n)$ rounds, or
\item using $O(\alpha)$ colors, in $O(\log n \cdot \min\{\log\log n,\;  \log \alpha\})$ rounds.
\end{itemize}
\smallskip
These algorithms are nearly-optimal, as it is known by results of Linial [FOCS'87] and Barenboim and Elkin [PODC'08]
that coloring with $\Theta(\alpha)$ colors, or even $\poly(\alpha)$ colors, requires $\Omega(\log_{\alpha} n)$ rounds. The previously best-known $O(\log n)$-time result
was a deterministic algorithm due to Barenboim and Elkin [PODC'08], which uses $\Theta(\alpha ^2)$ colors.
Barenboim and Elkin stated improving this number of colors as an open problem in their Distributed Graph Coloring Book.

\medskip

\end{abstract}

\newpage

\section{Introduction and Related Work}\label{sec:Intro}
\vspace{-7pt}
Graph coloring is one of the central and well-studied problems in \emph{distributed graph algorithms},
and it has a wide range of applications in networks and distributed systems, prototypically in scheduling conflicting tasks, e.g., transmission
in a wireless network. Much of the focus in the area has been on obtaining fast distributed algorithms that compute a $(\Delta+1)$-coloring, where $\Delta$ denotes the maximum degree of the graph,
see e.g.~\cite{ColeVishkin, linial1987LOCAL, panconesi1992improved, panconesirizzi2000, fraigniaud2016local, harris2016distributed, Schneider:2010, Schneider2011, barenboim2012locality, pettie2015distributed,
chung2014LLL, Szegedy:1993, Kuhn:2009:WGC, Barenboim:2009:DLC, Barenboim:2010:DDV, Barenboim:2015:D9S, Awerbuch:1989:NDL}. 

For a vast range of ``sparse'' graphs, using $\Delta+1$ colors is rather unsatisfactory. To take the point to the extreme,
coloring a tree---which is obviously $2$-colorable---using $\Delta+1$ colors seems quite wasteful. Generally, it is more
desirable to obtain colorings in which the number of colors improves if the graph is sparse (everywhere). 

In this paper, we present simple and near-optimal randomized distributed algorithms that compute a coloring of the graph with a number of colors that
depends on its (everywhere) sparsity, formally the \emph{arboricity} of the graph. We next review the related definitions and discuss the known results. Then, we state our contributions.
\medskip

\subsection{Definitions and Setup}
\vspace{-5pt}
\paragraph{Graph Arboricity} A standard measure of (everywhere) sparsity of an undirected graph $G=(V, E)$ is its \emph{arboricity},
defined as $$\alpha(G) = \max \bigg\{\big\lceil \frac{|E(V')|}{|V'|-1} \big\rceil  \; \bigg|\; V'\subseteq V, |V'|> 2\bigg\},$$ 
that is, roughly speaking, the maximum ratio of the number of edges to the number of vertices, among all subgraphs of $G$.
By a beautiful result of Nash-Williams\cite{nash1964decomposition}, an alternative equivalent formulation is as follows: arboricity $\alpha(G)$ is the minimum number of edge-disjoint forests to which one can partition the edges of $G$.

\medskip
\paragraph{The Distributed Model} As standard in distributed graph algorithms, we work with the \local model of
distributed computation\cite{linial1987LOCAL, Peleg:2000}: The network is abstracted as an undirected graph $G = (V,E)$, with $n =|V|$.
Communication happens in synchronous message-passing rounds, and per-round, each node can send one message to each of its
neighbors. We note that all of our algorithms work also in the more restricted variant of the model, known as \congest~\cite{Peleg:2000} model,
where each message can contain at most $O(\log n)$ bits. Initially, nodes do not know the topology of the graph, except for knowing the arboricity of the graph $\alpha(G)$.
At the end, each node should know its own part of the output, e.g., its own color in a coloring.  

\subsection{Known Results and Open Problems}
\vspace{-5pt}
\paragraph{Existential Aspects} Any graph $G$ admits a $2\alpha(G)$-coloring, and this bound is tight. For the former, note that one can easily arrange vertices as $v_1$, \dots, $v_n$
so that each $v_i$ has at most $2\alpha(G)-1$ neighbors $v_j$ with higher index $j>i$. Then, one can greedily color this list from $v_n$ to $v_1$, using $2\alpha(G)$ colors.
For the latter, note that a graph made of several disjoint cliques, each with $2\alpha$ vertices, has arboricity $\alpha$, and chromatic number $2\alpha$.

\medskip
\paragraph{Known Lower Bounds for Distributed Algorithms} By a classic observation of Linial\cite{linial1987LOCAL}, it is well-understood that having a small arboricity is
not a local characteristic of graphs, and any distributed algorithm for coloring with $2\alpha(G)$ colors, or anything remotely close to it, needs $\Omega(\log n)$ rounds. 
Concretely, Linial\cite{linial1987LOCAL} pointed out that there exists a graph with girth $\Omega(\log_{\Delta} n)$ and chromatic number $\Omega(\Delta/\log \Delta)$
\cite{bollobas1978chromatic}\footnote{In his original writing \cite{linial1987LOCAL},
Linial referred to such high-girth graphs with chromatic number $\Omega(\sqrt{\Delta})$, but he also added remarks that the bound can probably be improved to $\Omega(\Delta/\log \Delta)$.}
and thus also arboricity $\alpha=\Omega(\Delta/\log \Delta)$.
Graphs of girth $\Omega(\log_{\Delta} n)$ are indistinguishable from trees (which have arboricity $\alpha=1$), for distributed algorithms with round complexity $o(\log_{\Delta} n)$.
Hence, no distributed algorithm with round
complexity $o(\log_{\Delta} n)$ can compute a coloring of a tree with maximum degree $\Delta$---which clearly has arboricity $\alpha=1$---with less than $\Omega(\Delta/\log \Delta) \gg \poly(\alpha)$ colors. 

Barenboim and Elkin\cite{barenboim2008sublogarithmic, barenboim2010sublogarithmic, barenboim2013monograph} presented a strengthening of this result and showed that
for any $\alpha$ and $q < n^{1/4}/\alpha$, any distributed algorithm for $O(q \cdot \alpha)$-coloring graphs with arboricity $\alpha$ requires $\Omega(\log_{q\alpha} n)$ rounds.

\medskip
\paragraph{Known Distributed Algorithms for $(\Delta+1)$ Coloring}
Distributed graph coloring started with Linial's seminal work~\cite{linial1992locality,linial1987LOCAL}.
Linial’s coloring algorithm is an $O(\log^*n)$-round deterministic distributed
algorithm that computes an $O(\Delta^2)$-coloring of the input graph. This can be easily turned into a $\Delta+1$ coloring in $O(\Delta^2)$ additional rounds. In~\Cref{sec:Preliminaries:Linial}, we present a variation of
Linial's algorithm due to Barenboim and Elkin~\cite{barenboim2008sublogarithmic, barenboim2010sublogarithmic}, which
produces an $O(\alpha^2)$-coloring in $O(\log n)$ rounds of a graph $G$ with arboricity $\alpha$. Since Linial's algorithm, significant advances have been made in the area, which we briefly overview next. 

On the side of deterministic algorithms, the best known $(\Delta+1)$-coloring distributed algorithm,
in terms of dependency on $n$, is a $(2^{O(\sqrt{\log n})})$-round algorithm by Panconesi and Srinivasan~\cite{panconesi1992improved}.
In terms of dependency on the maximum degree $\Delta$ of the graph, the linear in $\Delta$ round complexity remained as the state of the art for
deterministic $\Delta+1$-coloring~\cite{barenboim2014distributed}, until very recently, when Barenboim~\cite{Barenboim:2015:D9S} presented an  $O(\Delta^{3/4}\log \Delta+ \log^*n)$-round distributed $(\Delta+1)$-coloring algorithm.
This was followed by a work of 
Fraigniaud, Heinrich, and Kosowski\cite{fraigniaud2016local}, which improved the round complexity to $O(\sqrt{\Delta} \log^{2.5} \Delta+ \log^*n)$ rounds.

On the side of randomized algorithms, an $O(\log n)$-round algorithm follows from Luby’s maximal
independent set (MIS) algorithm~\cite{luby1986simple}. A direct $O(\log n)$-round distributed algorithm was analyzed
by Johansson~\cite{johansson1999simple}. The fastest known randomized algorithm for
$(\Delta+1)$-coloring is due to a recent work of Harris et al.~\cite{harris2016distributed} which
provides a $(\Delta+1)$-coloring in $O(\sqrt{\log \Delta})+ 2^{O(\sqrt{\log \log n})}$ rounds, with high probability.

\paragraph{Shortcomings of These Methods in Obtaining Arboricity-Dependent Coloring} All the aforementioned deterministic and randomized algorithms perform in iterations, where in each iteration the graph
is colored partially and each node that remains uncolored removes from its palette the colors that are taken by its neighbors, until a proper $(\Delta+1)$-coloring of the whole graph is produced.
This fundamental property 
makes these algorithms inappropriate for our setting of obtaining an arboricity-dependent coloring of the graph.
In particular, in a graph $G$ with arboricity $\alpha$ and maximum out-degree $\Delta \gg \alpha$,
the above algorithms may fail to produce an $f(\alpha)$-coloring.
Next, we present the known results on distributed graph coloring in which the number of colors depends on the arboricity of the graph.
\smallskip

\paragraph{Known Distributed Algorithms for Arboricity-Dependent Coloring} Barenboim and Elkin\cite{barenboim2008sublogarithmic, barenboim2010sublogarithmic} present a deterministic distributed
algorithm that computes an $O(\alpha^2)$ coloring within $O(\log n)$ rounds --- which is essentially the time that is proven to be necessary by the above lower bound.
If one uses more colors, say $O(q \cdot \alpha^2)$ colors for some parameter $q\geq 1$, the algorithm can be made somewhat faster, running in $O(\log_{q} n + \log^* n)$ rounds.
They also show that by spending more time, particularly $O(\alpha \log n)$ rounds, one can get close to the ideal number of colors and use $\lfloor (2+\eps)\cdot \alpha +1 \rfloor$ colors,
for any constant $\eps>0$. This can be turned into smoother trade-off, obtaining an $O(t\cdot \alpha)$-coloring, for any $t\in [1, \alpha]$, in $O(\frac{\alpha}{t} \cdot \log n + \alpha \log \alpha)$ rounds.
\smallskip

Kothapoli and Pemmaraju\cite{Kothapalli:2011} study arboricity-dependent randomized distributed coloring algorithms, although targeting a very different range of parameters:
they allow drastically more colors, but then their algorithms run very fast. In particular, they present randomized distributed algorithms for $O(\alpha \cdot n^{1/k})$-coloring in $O(k)$ rounds,
when $k\in [\Omega(\log \log n, \sqrt{\log n})]$; see \cite[Theorem 1.4]{Kothapalli:2011} for the precise statements. They also present more detailed trade-offs in \cite[Theorem 1.3]{Kothapalli:2011},
when a low out-degree orientation of the graph is provided. By the above lower bounds, we know that if we want something remotely close to $2\alpha$ colors, or even $\poly(\alpha)$ colors,
we can allow $\Omega(\log_{\alpha} n)$ rounds for free. To the best of our understanding, the trade-offs of \cite[Theorem 1.4]{Kothapalli:2011} and \cite[Theorem 1.3]{Kothapalli:2011} are
not suitable when $\Omega(\log_{\alpha} n)$ rounds are allowed, with only one exception: for $\alpha \geq 2^{\omega(\log^{1/3} n)}$, one can obtain an $O(\alpha)$-coloring in $O(\log n)$ rounds,
by putting together \cite[Theorem 1.3 (ii)]{Kothapalli:2011} and $H$-partitions of \cite{barenboim2010sublogarithmic}.  

\medskip
\paragraph{Open Problem}
Barenboim and Elkin ask in Open Problem 11.11 of their distributed graph coloring book~\cite{barenboim2013monograph}: ``\emph{Can one use significantly less than $\alpha^2$ colors,
and still stay within deterministic $O(\log n)$ time?}'', immediately followed by adding that ``\emph{This question is open even for randomized algorithms}''.  

\subsection{Our Contribution}
We present very simple randomized distributed algorithms that make a significant progress on the above open problem:

\begin{theorem}\label{sec:Introduction and Related Work:theorem1}
For any constant $\eps>0$, there are randomized distributed algorithms that on any $n$-node graph with arboricity $\alpha$, with high probability\footnote{As standard, we use the phrase \emph{with high probability} (w.h.p.)
to indicate that an event happens with probability at least $1-1/n^c$, for a desirably large constant $c\geq 2$.}, compute 
\begin{itemize} 
\item a $\big(\min\big\{(2+\eps)\alpha + O(\log n\cdot \log\log n), O(\alpha\log \alpha)  \big\}\big)$-coloring in $O(\log n)$ rounds, 
\item an $O(\alpha)$-coloring, in $O(\log n \cdot \min\{\log\log n,\;  \log \alpha\})$ rounds.
\end{itemize}
\end{theorem}
\medskip

This theorem achieves a near-optimal coloring as a function of arboricity, with parameter trade-offs that compare favorably to the previous results provided by
\cite{Kothapalli:2011, barenboim2010sublogarithmic}. In particular, so long as $\alpha = \Omega(\log n \cdot \log \log n)$, we get the almost best-possible $((2+\eps)\cdot \alpha)$-coloring, for $\eps>0$, in $O(\log n)$ time.
For graphs of lower arboricity, we can either spend an $O(\log \alpha) \leq O(\log\log n)$ factor more time and get an $O(\alpha)$-coloring in $O(\log n \cdot \log \log n)$ rounds, or we can use an $O(\log \alpha)$
factor more colors and get a coloring with $O(\alpha \log \alpha) \ll \alpha^2$ colors in $O(\log n)$ time.

\section{Warm Up: Reviewing an Algorithm of Barenboim and Elkin\cite{barenboim2008sublogarithmic, barenboim2010sublogarithmic}}
In this section, we review an $O(\log n)$-round deterministic distributed algorithm by Barenboim and Elkin~\cite{barenboim2010sublogarithmic}
that produces an $O(\alpha^2)$-coloring of any $n$-node graph $G=(V,E)$ with arboricity $\alpha$.
We note that the paper \cite{barenboim2010sublogarithmic} presents other trade-offs when more time is allowed, as overviewed in
\Cref{sec:Intro}, e.g., $((2+\eps)\cdot \alpha)$-coloring in $O(\alpha \log n)$ time, but these algorithms are less relevant for our target of $O(\log n)$-time algorithms
(and also their aforementioned open problem in \cite{barenboim2013monograph}).

The $O(\log n)$-time $O(\alpha^2)$-coloring algorithm of Barenboim and Elkin\cite{barenboim2008sublogarithmic, barenboim2010sublogarithmic} consists of two steps.
In the first step, we use an algorithm, called $H$-partition, to compute an orientation of the edges in $O(\log n)$ rounds,
such that each node has out-degree at most $O(\alpha)$. In the second step, we compute an $O(\alpha^2)$-coloring in $O(\log^* n)$ rounds,
using the low out-degree orientation of step 1. Later in \Cref{sec:Coloring for High-Arboricity Graphs,sec:Coloring for Low-Arboricity Graphs}, we will make use of this $H$-partition method.

\subsection{Step 1: Low Out-Degree Orientation via $H$-partition}\label{sec:Preliminaries:orientation}
We now discuss a deterministic distributed algorithm that, given an $n$-node graph $G=(V,E)$ with arboricity $\alpha$, in $O(\log_{1+\eps/2} n)$ rounds, 
computes an acyclic orientation of the edges such that the maximum out-degree is at most $(2+\eps)\cdot \alpha$, for a given parameter $\eps>0$. 

The main idea behind the algorithm is to partition the nodes into
$\ell=\lceil \log_{\frac{2+\eps}{2}} n \rceil$ disjoint subsets 
$H_1,H_2,...H_{\ell}$,
such that every node $v \in H_j$ with $j \in \{1,2,...,l\}$, has at most $(2+\eps)\cdot \alpha$ neighbors in subsets $\cup_{y=j}^{\ell} H_j$.
We refer to partitions that satisfy this property as $H$-partitions with degree $d\le (2+\eps)\cdot \alpha$ and size $\ell=\lceil \log_{\frac{2+\eps}{2}} n \rceil$.
We refer to subsets
$H_1,H_2,...,H_{\ell}$ as layers of the $H$-partition. In \Cref{sec:Preliminaries:H-partition:algorithm} we sketch the algorithm for
computing an $H$-partition.

Once an $H$-partition is computed, we orient the edges that have endpoints in different layers $H_j$ and $H_j'$, for $j'>j$, 
towards the higher layer $H_j'$, and orient the edges which have endpoints in the same layer towards the greater ID endpoint.
This ensures that we have an acyclic orientation with maximum out-degree
at most $ d \le (2+\eps)\cdot \alpha$. 

\begin{lemma}\label[lemma]{sec:Preliminaries:H-partition:algorithm} 
For a graph $G$ with arboricity $\alpha$ and a parameter $\eps>0$,
there is a deterministic distributed algorithm that computes
an $H$-partition of $G$ with  degree $d \le (2+\eps) \cdot \alpha$ and size $\ell=\lceil \log_{\frac{2+\eps}{2}} n \rceil$
in $O(\log_{\frac{2+\eps}{2}} n )$ rounds. 
\end{lemma}

\begin{proof}[Proof Sketch]
A graph with arboricity $\alpha$ has at least $\frac{\eps}{2+\eps}\cdot n$ nodes
with degree at most $(2+\eps)\cdot \alpha$, as can be seen by a simple double-counting of edges. These nodes join layer $H_1$.
In the subgraph $G\setminus H_1$, there are at least
$\frac{\eps}{2+\eps}\cdot (n-|V(H_1)|)$ nodes with degree at most $(2+\eps)\cdot \alpha$.
These nodes join layer $H_2$. Iteratively, in the subgraph $G\setminus \cup_{y=1}^j H_y$
there are  at least
$\frac{\eps}{2+\eps}\cdot (n-\sum_{y=1}^{j}|V(H_y)|)$ nodes with degree at most $(2+\eps)\cdot \alpha$; these nodes
join layer $H_{j+1}$.
This argument continues until all nodes have joined a layer, which happens after at most
$\ell=\lceil \log_{\frac{2+\eps}{2}} n \rceil$ rounds. 
\end{proof}

\subsection{Step 2: Coloring the Graph using the Low Out-Degree Orientation}\label[section]{sec:Preliminaries:Linial}
We now employ the above low out-degree (acyclic) orientation to compute an $O(\alpha^2)$-coloring, in $O(\log^* n)$ additional rounds. The algorithm 
is based on (iterative applications of) a single-round coloring reduction,
similar to Linial's Algorithm~\cite{linial1987LOCAL,linial1992locality}. 

\subsubsection*{Linial's Coloring Algorithm}
Linial's coloring algorithm is an
$O(\log^*n)$-round 
deterministic distributed algorithm that computes an $O(\Delta^2)$-coloring of the input graph,
where $\Delta$ is the largest degree in the graph.
In each round, a $k$-coloring is transformed to a $k'$-coloring, such that
$k'=O(\Delta^2 \log_{\Delta} k )$. This is done by letting each node
compute a set that is not a subset of the union of the sets of its neighbors. 
Then, it picks an arbitrary color from this set
that is not in the union of its neighbors' sets. 
The existence of such a set relies on \Cref{sec:Preliminaries:Linial:cover-free}.
The coloring is produced by iteratively applying the single-round color reduction.
We start with the initial numbering of the vertices as a $n$-coloring.
In a single round, we compute an $O(\Delta^2 \log_{\Delta} n)$-coloring.
With another single-round color reduction, we get an $O (\Delta^2\cdot (\log_{\Delta} \Delta + \log_{\Delta} \log_{\Delta} n))$ coloring. 
After $O(\log^*n)$ iterations, we end up with an $O(\Delta^2)$-coloring. The single-round reduction technique relies on the following lemma.

\begin{lemma}\label[lemma]{sec:Preliminaries:Linial:cover-free}[Linial \cite{linial1987LOCAL,linial1992locality}]
For any $k$ and $\Delta$, there exists a $\Delta$-cover free family of size $k$ on a ground-set of size $k'=O(\Delta^2 \log_\Delta k)$
i.e., a family of sets $S_1,S_2,...,S_k \in \{1,2,...,k'\}$ such that there is no set in the family that is a subset of the union of $\Delta$ other sets.
\end{lemma}

\subsubsection*{Applying Linial's Algorithm to Low Out-Degree Graphs}
Now, the second step of the $O(\alpha^2)$-coloring algorithm of Barenboim and Elkin\cite{barenboim2008sublogarithmic, barenboim2010sublogarithmic}
is running a variation of Linial's algorithm where each node considers only the colors of its out-neighbors.
In particular, each node computes a set that is not a subset of the union of the sets of its out-neighbors.
Then, it picks an arbitrary color from this set that is not in the union of its out-neighbors' sets. This produces a proper coloring of the graph.
Similar to Linial's algorithm, after $O(\log^*n)$ rounds, the number of colors is $O(\alpha^2)$.

\section{Coloring for High-Arboricity Graphs}\label{sec:Coloring for High-Arboricity Graphs}
In this section, we present an $O(\log n)$-round randomized distributed algorithm that, with high probability
computes, a $((2+\varepsilon)\cdot \alpha+O(\log n \cdot \log\log n))$-coloring of a graph $G$ with arboricity $a=\Omega(\log n)$, for any desirably small constant $0 < \varepsilon \le 1$.

\paragraph{Algorithm Outline}
Our algorithm consists of two steps. 

\begin{itemize} 
\item In the first step, we perform an $O(\log n)$-round partial coloring that uses $(2+\frac{2}{3} \varepsilon)\cdot \alpha$ colors, in a manner that
the remaining graph--- i.e., the graph induced by the nodes that remain uncolored --- has arboricity at most $\frac{\varepsilon}{144}\alpha$, with high probability.

\item In the second step, we partially color the remaining graph of arboricity at most $\frac{\varepsilon}{144}\alpha$, in $O(\log n)$ rounds, using at most $\frac{\varepsilon}{3}\alpha$ new colors.
This is done such that at the end of the second step, the subgraph induced by the uncolored nodes has arboricity at most $O(\log n)$, with high probability. 
\end{itemize}
Overall, our algorithm runs in $O(\log n)$ rounds and uses $(2+\varepsilon)\cdot \alpha$ colors.
Once we are done with this 2-step partial coloring,
on the remaining graph, we apply the coloring algorithm of \Cref{lem:coloringforLowArb}, which we present later in \Cref{sec:Coloring for Low-Arboricity Graphs}.
This algorithm uses $O(\log n\cdot \log\log n)$ new colors to color the remaining uncolored nodes, in $O(\log n)$ rounds. Hence, overall,
we obtain a proper $((2+\varepsilon)\cdot \alpha+O(\log n \cdot \log\log n))$-coloring of the whole graph, in $O(\log n)$ rounds, with high probability.
If we omit the first step and apply directly the second step of the algorithm, an $O(\alpha)$ partial coloring is produced in $O(\log n)$ rounds.
Overall, this would produce a proper $O(\alpha)$-coloring of the the whole graph, in $O(\log n)$ rounds, with high probability.

We note that if the input graph $G$ has arboricity $\alpha \geq \log^2 n$, once we reach a remaining graph of arboricity $O(\log n)$,
we can wrap up using a much simpler algorithm: we can color the remaining graph by applying the variation of Linial's algorithm explained in
\Cref{sec:Preliminaries:Linial}, which uses $O(\log^2 n)$ extra colors and colors all the remaining nodes in $O(\log^*n)$ extra rounds. Hence, in total,
we would end up with a $((2+\eps)\cdot \alpha + O(\log^2 n))$-coloring in $O(\log n)$ rounds.

\subsection{Step 1: A First Partial Coloring of the Graph}
Let $G=(V,E)$ be a graph with arboricity $\alpha=\Omega(\log n)$. In this section,
we present an $O(\log n)$-round randomized distributed algorithm that partially colors $G$,
using $(2+\frac{2\varepsilon}{3})\cdot \alpha$ colors, for a small constant $0 < \varepsilon \le 1$,
such that the remaining graph i.e., the graph induced by the remaining uncolored nodes, has 
arboricity at most $\frac{\varepsilon}{144}\alpha$. Next, for simplifying the notation, we use $\epsilon=\frac{\varepsilon}{3}$.

A first preparation step of the algorithm is to compute in $O(\log n)$ rounds an $H$-partition with degree $d \le (2+\epsilon)\cdot \alpha$ and size $\ell=\lceil \frac{1}{\epsilon}\log n\rceil$,
together with an acyclic orientation of the edges, such that the maximum out-degree is at most $d \le (2+\epsilon)\cdot \alpha$.
Then, it partially colors layers $H_1, H_2,...,H_{\ell}$ gradually, starting from layer $H_{\ell}$ and proceeds backwards, ending with the first layer $H_1$. 
Each node receives a palette of size $(2+2\epsilon)\cdot \alpha$ and when we color layer $H_j$, $1 \le j \le \ell$, each (uncolored) node $v \in H_j$ performs the following algorithm.

\begin{framed}\label{sec:Coloring for High-Arboricity Graphs:step1:randomized algorithm}
{\setlength{\parindent}{0cm}
\textbf{First Random Partial Coloring Algorithm, run by each node $v \in H_j$: \\}
In iteration $i\in \{1,2,...,\lceil \frac{1+\epsilon}{\epsilon}\rceil \cdot \log \frac{300}{\epsilon} \}$,
\begin{itemize}
 \item Node $v$ selects one random color $x$ among colors $\{1,2,...,(2+2\epsilon)\cdot \alpha\}$.
 \item Node $v$ sends the selected color $x$ to its neighbors, and receives their selected colors.
 \item If no out-neighbor has selected $x$ in this round, or picked $x$ as its permanent color in the previous rounds, node $v$ gets colored permanently with $x$, and informs its neighbors.
\end{itemize}}
\end{framed}
\begin{lemma}
After partially coloring the graph in $O(\log n)$ rounds, the remaining graph i.e., the graph induced by the uncolored nodes, has 
out-degree at most $\frac{\epsilon}{112}d$, with high probability.
\end{lemma}

\begin{proof}
First, we discuss the time complexity of the algorithm. We have $\lceil \frac{1+\epsilon}{\epsilon}\rceil \cdot \log \frac{300}{\epsilon} $ iterations per layer of the
$H$-partition and the $H$-partition has $ \ell=\lceil \frac{1}{\epsilon}\log n\rceil$ total layers. Hence, the whole algorithm
has round complexity $O(\log n)$.

We now argue that once the algorithm is completed, with high probability, the remaining graph has arboricity at most $\frac{\epsilon}{112}d$.
Consider an arbitrary layer $H_j$, $1 \le j \le \ell$ of the $H$-partition.
A node $v \in H_j$ has at most $d \le (2+\epsilon)\cdot \alpha$ neighbors in the graph induced by layers $\cup_{y=j}^{\ell} H_y$.
In each iteration $i$, each permanently colored out-neighbor of $v$, blocks at most one color from $v$'s palette.
Each out-neighbor that is in the same
layer $H_j$ and remains uncolored in iteration $i$,
blocks at most 1 color from $v$'s palette in iteration $i$. This implies that in any iteration $i$,
$v$ has at least $\epsilon \cdot \alpha$ colors that are not blocked by its out-neighbors.
Therefore, the probability that $v$ gets 
permanently colored with a color $x$ in iteration $i$ is at least $\frac{\epsilon \cdot \alpha}{(2+2\epsilon)\cdot \alpha}$.
Moreover, this holds independently of the events of other nodes being colored.

In total, after $\lceil \frac{1+\epsilon}{\epsilon}\rceil \cdot \log \frac{300}{\epsilon}$ iterations we get that, independently of the events of other nodes being colored,
$$Pr[\text{v is not colored}] \le (1-\frac{\epsilon \cdot \alpha}{(2+2\epsilon)\cdot \alpha})^{\lceil \frac{1+\epsilon}{\epsilon}\rceil \cdot \log \frac{300}{\epsilon}}
\le (\frac{1}{4})^{\frac{1}{2} \log \frac{300}{\epsilon}} \le \frac{\epsilon}{300}.$$

After applying the partial coloring in layers $H_1, H_2,..,H_{\ell}$, each node remains uncolored with probability 
at most $\frac{\epsilon}{300}$. 

At this point, the coloring process of the algorithm is completed. We now upper bound the arboricity of the remaining graph i.e., the
graph induced by the uncolored nodes after applying the algorithm. 
Consider a node $v$ that remains uncolored and 
let $X$ be a random variable that represents the number
of $v$'s uncolored out-neighbors. Then, $$E[X]\le d \cdot \frac{\epsilon}{300}.$$
So long as the expected out-degree is $\Omega(\log n)$, we can apply the Chernoff bound and conclude that
$$Pr[X \ge d \cdot \frac{\epsilon}{112}]\leq \frac{1}{n^{10}}.$$
Hence, the remaining graph is an $H$-partition with degree $d \le\frac{\epsilon}{112}(2+\epsilon)\alpha\le \frac{\eps}{336}(2+\frac{\eps}{3})\alpha \le \frac{\eps}{144}\alpha$ and size $\ell=\lceil \frac{3}{\eps}\log n\rceil$
and is oriented such that the out-degree of each remaining node is at most $d \le \frac{\eps}{144}\alpha$, with high probability.
\end{proof}

\subsection{Step 2: A Second Partial Coloring of the Remaining Graph}
Once the first step of the algorithm is completed, the remaining graph is
an $H$-partition with degree $d \le \frac{\eps}{144}\alpha$ and size $\ell=\lceil \frac{3}{\epsilon}\log n\rceil$ and is oriented such that the
out-degree is at most $d \le \frac{\eps}{144}\alpha$ , with high probability.

In this section, we present an $O(\log n)$ randomized distributed algorithm that partially color this remaining graph using $48d \le \frac{\eps}{3}\alpha$ colors,
in a manner that once the algorithm is completed, the graph induced by the remaining uncolored nodes has arboricity at most $O(\log n)$, 
with high probability.

\begin{lemma}
Given an $H$-partition with degree $d=\Omega(\log n)$ and
size $O(\log n)$,
there is an $O(\log n)$ randomized distributed algorithm that partially colors the graph 
using $48d$ colors,
in a manner that the remaining graph has arboricity at most $O(\log n)$, 
with high probability.
\end{lemma}

\begin{proof}
The algorithm consists of $\log^* n$ phases. In each phase $i$, for $i\in \{0,1,...,\log^* n\}$, we perform a partial coloring
of the remaining graph as follows. The input of phase $i$ is an $H$-partition of the remaining graph
with degree $d_i \le \frac{d}{{^{i}2}}$
and size $O(\frac{\log n}{2^i})$. 
Here, the tetration ${^{y}x}$ expresses $x^{x^{.^{.^{x} } } }$, with $y$ copies of $x$.
In each phase $i$, we apply the $O(\frac{\log n}{2^i})$-round randomized distributed algorithm of~\Cref{sec:Coloring for High-Arboricity Graphs:lemma:one phase},
which we discuss later on this section. In phase $i$, we use $2Q_i=2\cdot  \frac{12d}{2^i}$ colors and we partially color the graph such that
at the end of phase $i$, the remaining graph is an $H$-partition
with degree $d_{i+1} \le \frac{d}{{^{(i+1)}2}}$ and size $O(\frac{\log n}{2^{i+1}})$, with high probability. This is the input for the next phase.

After $\log^* n$ phases, the remaining nodes have out-degree at most $O(\log n)$, with high probability. Furthermore, the total number of
rounds of the process is $\sum_{i=0}^{\log^* n} \frac{O(\log n)}{2^i}=O(\log n)$ and the total number of colors that it uses is $\sum_{i=0}^{\log^* n} 2Q_i\le 48d$.
\end{proof}

\paragraph{The Coloring Algorithm for a Single Phase}
For each phase $i$, we start with an $H$-partition of the remaining graph with degree $d_i \le \frac{d}{{^{i}2}}$ and size $O(\frac{\log n}{2^i})$.
In the coloring part of this phase, we color some nodes in a manner that, among the nodes that remain uncolored,
each node has out-degree at most $\frac{d}{{^{(i+2)}1.98\cdot 20}} \ll \frac{d}{{^{i}2}}$.

The coloring process in phase $i$ consists of two iterations, as follows: In each iteration, each remaining node receives a fresh palette of $Q_i=\frac{12d}{2^i}$ colors.
We color the layers $H_1, H_2,..,H_{\ell}$ of the given $H$-partition gradually, starting from the last layer $H_{\ell}$, and proceed backwards, ending with the first layer $H_1$.
As we show next, one iteration is not enough to drop the maximum out-degree to the desired level. Repeating the algorithm for a second iteration,
we end up with maximum out-degree at most $\frac{d}{{^{(i+2)}1.98\cdot 20}} \ll \frac{d}{{^{i}2}}$, with high probability.
We now focus on coloring an arbitrary layer $H_j$, $1 \le j \le \ell$.
Each node $v$ in layer $H_j$ performs the following algorithm.

\begin{framed}\label{sec:Coloring for High-Arboricity Graphs:algorithm}
{\setlength{\parindent}{0cm}
\textbf{Single-Iteration of Second Partial Coloring Algorithm, run by each node $v\in H_j$ }
\begin{itemize}
 \item Node $v$ selects $f(i)=\frac{Q_i}{2d_i}$ colors at random from a new palette of $Q_i$ colors.
 \item Node $v$ sends the selected colors to its neighbors, and receives their selected colors.
 \item If there is a selected color $x$ such that no out-neighbor has selected $x$ in this round, or picked $x$ as its permanent color
in the previous rounds, node $v$ gets colored permanently with $x$, and informs its
neighbors.
\end{itemize}}
\end{framed}

\begin{lemma}\label[lemma]{sec:Coloring for High-Arboricity Graphs:lemma:one phase}
 Given an $H$-partition with degree $d_i \le \frac{d}{{^{i}2}}$ and size $O(\frac{\log n}{ 2^i})$,
 there is an $O(\frac{\log n}{ 2^i})$-round randomized distributed algorithm that partially colors the graph with
 $2 Q_i=2\cdot \frac{12d}{2^i}$ colors, such that in the same $H$-partition, with size  $O(\frac{\log n}{2^{i}})$, the remaining graph has out-degree
 at most $\frac{d}{{^{(i+2)}1.98\cdot 20}}$, with high probability.
\end{lemma}

\begin{proof}
First, we discuss the time complexity of the algorithm. We have two iterations, and each iteration takes $\ell=O(\frac{\log n}{2^i})$ rounds, one round per layer of the $H$-partition.
Hence, the whole algorithm of this phase has round complexity $\ell=O(\frac{\log n}{2^i})$.

We now argue that at the end of the phase, with high probability, in the remaining graph induced by the uncolored nodes each node has out-degree at most $\frac{d}{{^{(i+2)}1.98\cdot 20}}$.
We do the analysis of the two iterations separately, though they are similar.

Consider the first iteration of phase $i$ and an arbitrary layer $H_j$, $1 \le j \le \ell$. A node $v \in H_j$ has at most $d_i$ out-neighbors in the graph induced by
layers $\cup_{y=j}^{\ell} H_y$. Each permanently colored out-neighbor of $v$ blocks at most one color from $v$'s palette. Each out-neighbor that belongs to the same
layer $H_j$, blocks at most $f(i)$ colors from $v$'s palette.

Thus, there are at most $f(i)\cdot d_i$ colors that are blocked by $v$'s out-neighbors, which implies that $v$
has at least $Q_i-f(i)\cdot d_i=\frac{Q_i}{2}$ colors that are not blocked, when we select random colors for $v$. Therefore, the probability that $v$ gets 
permanently colored with a color $x$ that it selects is at least 1/2. Moreover, this holds independently of the events of other nodes being colored.
In total, since $v$ selects $f(i)=\frac{Q_i}{2d_i}$ colors independently, we get that independently of the events of other nodes being colored:
$$Pr[\text{v is not colored}] \le 2^{-f(i)}=2^{-\frac{{^{i}2\cdot 6}}{ 2^i}}.$$

After applying the 1-round coloring in layers $H_1, H_2,..,H_{\ell}$, each node remains uncolored with probability 
at most $2^{-f(i)}$. 

At this point, the coloring process of the first iteration is completed. We now upper bound the maximum out-degree of the remaining graph.
Consider a node $v$ that remains uncolored and 
let $X$ be a random variable that represents the number
of $v$'s uncolored out-neighbors. Then, $$E[X]\le d_i \cdot 2^{-f(i)}\le \frac{d}{{^{i}2}} \cdot 2^{-\frac{{^{i}2\cdot 6}}{2^i}}.$$
As long as the new expected out-degree is $\Omega(\log n)$, we can apply the Chernoff bound and conclude that
$$Pr[X \ge \frac{d}{{^{(i+1)}1.99 \cdot 20}}]\leq Pr[X \ge 3 \frac{d}{{^{i}2 \cdot 64}}{2}^{-\frac{{^{i}2}}{2^i}}] \le \frac{1}{n^{10}}.$$

We now discuss the decrease in the out-degrees during the second iteration. At the beginning of the second iteration, in the remaining graph, each (uncolored) node has at most $\frac{d}{{^{(i+1)}1.99 \cdot 20}}$ out-neighbors, 
with high probability. %
Similarly to the first iteration, each remaining node receives a fresh palette of size $Q_i$. 
Again, applying the same process, after we color layers $H_1, H_2,..,H_{\ell}$ in the second iteration, each node remains uncolored with probability 
at most $2^{-f(i)}$. With a similar analysis, we conclude that in the graph induced by nodes that remain uncolored at the end of the second iteration, each node has out-degree at most
$\frac{d}{{^{(i+2)}1.98\cdot 20}}$, with high probability.
\end{proof}

\paragraph{Re-computing the $H$-partition}
At this point, we are done with the coloring of phase $i$. As a preparation step for phase $i+1$, 
we compute a new $H$-partition of the graph induced by the uncolored nodes. 
The new $H$-partition has degree $d_{i+1} \le \frac{d}{{^{(i+1)}2}}$ and size $O(\frac{\log n}{2^{i+1}})$.

\begin{lemma}\label[lemma]{sec:Coloring for High-Arboricity Graphs:lemma:h-partition}
 Given an $H$-partition with 
 degree at most $\frac{d}{{^{(i+2)}1.98\cdot 20}}$ and size $O(\frac{\log n}{2^{i}})$,
 there is an $O(\frac{\log n}{2^{i+1}})$-round deterministic distributed algorithm that 
 computes an $H$-partition with 
 degree at most $\frac{d}{{^{(i+1)}2}}$ and size $O(\frac{\log n}{2^{i+1}})$.
\end{lemma}

\begin{proof}
We set the parameter $\eps>0$ of the $H$-partition of~\Cref{sec:Preliminaries:H-partition:algorithm}, to a value such that the degree of the $H$-partition is
$(2+\eps )\frac{d}{{^{(i+2)}1.98\cdot 20}}\le 
\frac{d}{^{(i+1)}2}$
and the size of the $H$-partition is $\ell=\frac{\log n}{\log \eps} \le \frac{\log n}{2^{i+1}}.$
In particular, we set $\eps=16\frac{^{(i+2)}1.98}{^{(i+1)}2}$, and compute an $H$-partition with degree
$d_{i+1} 
\le 
\frac{d}{^{(i+1)}2}$
and size $\ell \le \frac{\log n}{2^{i+1}}$. The round complexity of recomputing the $H$-partition is at most $O(\frac{\log n}{2^{i+1}})$,
as explained in~\Cref{sec:Preliminaries:H-partition:algorithm}.
\end{proof}

\section{Coloring for Low-Arboricity Graphs}\label{sec:Coloring for Low-Arboricity Graphs}
In this section, we present two randomized distributed algorithms that on
any $n$-node graph with arboricity $\alpha$, with high probability, compute respectively
\begin{itemize}
 \item an $O(\alpha \log \alpha )$-coloring in $O(\log n)$ rounds, and
 \item an $O(\alpha)$-coloring in $O(\log n \cdot \log \alpha)$ rounds.
\end{itemize} 
In particular, we prove the following two lemmas in~\Cref{sec:Low-Arboriciry Partial Coloring} and~\Cref{sec: Low-Arboricity Coloring:O(a) algorithm}, respectively.

\begin{lemma}\label[lemma]{lem:coloringforLowArb} There is an $O(\log n)$-round randomized distributed algorithm that partially colors any $n$-node graph with arboricity $\alpha$,
using $O(\alpha \log \alpha)$ colors,
in a manner that the remaining graph has no path longer than $O(\log n)$, with high probability.
\end{lemma}

\begin{lemma} \label[lemma]{lem:coloringforLowArb2}
There is an $O(\log n \cdot \log \alpha)$-round randomized distributed algorithm that partially colors any $n$-node graph with arboricity
$\alpha$, using $(2+\eps)\cdot \alpha$ colors, for a constant $0<\eps \le1$, in a manner that the remaining graph has no path longer than $O(\log n)$, with high probability.
\end{lemma}

After partially coloring the graph with the algorithms of~\Cref{lem:coloringforLowArb} or~\Cref{lem:coloringforLowArb2},
we apply the $O(\log n)$-round deterministic distributed algorithm of~\Cref{sec:Low-Arb-Coloring:deterministic-coloring},
to color the remaining graph using $O(\alpha)$ extra colors.

We note that the algorithms we present in this section are more interesting for coloring graphs with arboricity at most $O(\log n)$,
since for graphs with larger arboricity, we can apply the algorithm of~\Cref{sec:Coloring for High-Arboricity Graphs} to obtain a $((2+\eps)\cdot \alpha + O(\log n \cdot \log \log n))$-coloring
in $O(\log n)$ rounds.

\subsection{A Randomized $O(\alpha \log \alpha)$ Partial Coloring in $O(\log n)$ rounds}\label[section]{sec:Low-Arboriciry Partial Coloring}
Let $G$ be a $n$-node graph with arboricity $\alpha$. In this section, we provide an $O(\log n)$-round
randomized distributed algorithm that partially colors the graph with $O(\alpha \log \alpha)$ colors,
in a manner that the remaining graph has no path longer than $O(\log n)$, with high probability.

A first preparation step of the algorithm is to compute in $O(\log n)$ rounds an $H$-partition
with degree $d \le 3\alpha$ and size $O(\log n)$, together with an acyclic orientation of the edges, such
that the maximum out-degree is at most $d \le 3\alpha$.

The algorithm colors layers $H_1, H_2,...,H_{\ell}$ gradually, starting from layer $H_{\ell}$ and proceeds backwards, ending with the first layer $H_1$. Initially,
each node receives a palette of $d\log d$ colors.
When layer $H_j$, $1 \le j \le \ell$ is colored, each remaining node $v \in H_j$ performs the following algorithm.

\begin{framed}\label{sec:Coloring for Low-Arboricity Graphs:randomized algorithm}
{\setlength{\parindent}{0cm}
\textbf{Low-Arb Coloring Algorithm, run by each node $v\in H_j$ \\}
In iteration $i \in \{1,2,3,4\}$:
\begin{itemize}
 \item Node $v$ selects $\frac{\log d }{2}$ random colors among $d \log d$ colors.
 \item Node $v$ sends the selected colors to the neighbors, and receives their selected colors.
 \item If there is a selected color $x$ such that no out-neighbor has selected $x$ in this round, or picked $x$ as its permanent color
in the previous rounds, node $v$ gets colored permanently with $x$, and informs its
neighbors.
\end{itemize}}
\end{framed}

\begin{lemma}\label[lemma]{sec:Coloring for Low-Arboricity Graphs:probability for one color}
After partially coloring the graph in $O(\log n)$ rounds, each node $v \in V$ remains uncolored with probability at most $d^{-2}$.
Furthermore, this holds independently of the events of other nodes being colored.

\end{lemma}

\begin{proof}
First, we discuss the time complexity of the algorithm. The $H$-partition has $O(\log n)$ layers and
we have 4 iterations per layer of the $H$-partition.
Hence, the whole algorithm has round complexity $O(\log n)$.

We now argue that once the algorithm is completed, in the remaining graph, each node $v \in V$ remains uncolored with probability at most $d^{-2}$,
independently of the events of other nodes being colored in the graph.

Consider an arbitrary layer $H_j$, $1 \le j \le \ell$ of the $H$-partition.
A node $v \in H_j$ has at most $d$ out-neighbors in the graph induced by layers $\cup_{y=j}^{\ell} H_y$.
In each iteration $i$, each permanently colored out-neighbor of $v$ blocks at most one color from $v$'s palette.
Each out-neighbor that is in the same
layer $H_j$ and remains uncolored in iteration $i$,
blocks at most $\frac{\log d}{2}$ colors from $v$'s palette.

Thus, $v$ has at least $\frac{d \log d}{2}$ colors that are not blocked by its out-neighbors, when we select random colors for $v$. Therefore, the probability that $v$ gets 
permanently colored with a color $x$ that it selects in iteration $i$ is at least 1/2. Moreover, this holds independently of the events of other nodes being colored. This implies
that in each iteration, independently of the events of other nodes being colored, we have
$$Pr[\text{v is not colored}] \le 2^{-\frac{\log d }{2}} = 1/\sqrt{d}.$$

In total, after 4 iterations we get that, independently of the events of other nodes being colored, we have
$$Pr[\text{v is not colored}] \le (1/\sqrt{d})^{4}=d^{-2}.$$
\end{proof}

Next, we prove that in the remaining graph, there exists no path longer than $O(\log n)$, with high probability.
This allows us to color the remaining graph deterministically in $O(\log n)$ rounds, using $d+1$ extra colors, as we explain in~\Cref{sec:Low-Arb-Coloring:deterministic-coloring}.

\begin{lemma}\label[lemma]{sec:Coloring for Low-Arboricity Graphs:lemma:path length}
The remaining graph has no directed path longer than $O(\log n)$, w.h.p.
\end{lemma}

\begin{proof}
There are at most $n \cdot d^{\log n}$ different ways to select a path of length ${\log n}$.
For each such path, the probability that all of its nodes stay is at most
$d^{-2{\log n}}$. 
By a union bound over all such paths, we conclude that with probability
$1-n\cdot d^{\log n} \cdot d^{-2\log n}\ge 1- n^{-10}$, no such path exists.
\end{proof}

\subsection{A Randomized $O(\alpha)$ Partial Coloring in $O(\log n \cdot \log \alpha)$ Rounds}\label[section]{sec: Low-Arboricity Coloring:O(a) algorithm}
In this section, we present an $O(\log n \cdot \log \alpha)$-round randomized distributed algorithm that colors a graph $G$ with arboricity $\alpha$,
using $(2+\eps)$ colors, for a small constant $0< \eps \le 1$, in a manner that the remaining graph has no path longer than $O(\log n)$, with high probability.

The algorithm is similar to the randomized distributed algorithm of \Cref{sec:Low-Arboriciry Partial Coloring}. 
More specifically, it first computes an $H$-partition with degree $d \le (2+\frac{\eps}{2})\cdot \alpha$ and size $ \ell=\lceil \frac{1}{\eps}\log n\rceil$.
Each node receives a palette of size $(2+\eps)\cdot\alpha$ and when we color layer $H_j$, $1 \le j \le \ell$, each (uncolored) node performs the following algorithm.

\begin{framed}
{\setlength{\parindent}{0cm}
\textbf{Tradeoff-Low-Arb Coloring Algorithm, run by each node $v \in H_j$: \\}
In iteration $i\in \{1,2,...,\lceil \frac{2\cdot(2+\eps)}{\eps}\rceil \cdot \log d \}$,
\begin{itemize}
 \item Node $v$ selects one random color $x$ among $(2+\eps)\cdot \alpha$ colors.
 \item Node $v$ sends the selected color $x$ to its neighbors, and receives their selected colors.
 \item If no out-neighbor has selected $x$ in this round, or picked $x$ as its permanent color in the previous rounds, node $v$ gets colored permanently with $x$, and informs its neighbors.
\end{itemize}}
\end{framed}

\begin{lemma}\label[lemma]{sec:Coloring for Low-Arboricity Graphs:tradeoffalgorithm}
After partially coloring the graph in $O(\log n \cdot \log \alpha)$ rounds, each node $v \in V$ remains uncolored with probability at most $d^{-2}$.
Furthermore, this holds independently of the events of other nodes being colored.
\end{lemma}

\begin{proof}
First, we discuss the time complexity of the algorithm. The $H$-partition has $O(\log n)$ layers and
we have $\lceil \frac{2\cdot(2+\eps)}{\eps}\rceil \cdot \log d=O(\log \alpha)$ iterations per layer of the $H$-partition.
Hence, the whole algorithm has round complexity $O(\log n \cdot \log \alpha)$.

We now argue that once the algorithm is completed, in the remaining graph, each node $v \in V$ remains uncolored with probability at most $d^{-2}$,
independently of the events of other nodes being colored in the graph.

Consider an arbitrary layer $H_j$, $1 \le j \le \ell$ of the $H$-partition.
A node $v \in H_j$ has at most $d \le (2+\frac{\eps}{2})\cdot \alpha$ neighbors in the graph induced by layers $\cup_{y=j}^{\ell} H_y$.
In any iteration $i$, each permanently colored out-neighbor of $v$, blocks at most one color from $v$'s palette.
Each out-neighbor that is in the same
layer $H_j$ and remains uncolored in iteration $i$,
blocks at most one color from $v$'s palette. Thus, in any iteration $i$, node
$v$ has at least $\frac{\eps}{2} \alpha$ colors that are not blocked by its out-neighbors.
Therefore, the probability that $v$ gets 
permanently colored with a color $x$ in iteration $i$ is at least $\frac{\eps \cdot \alpha}{2(2+\eps)\cdot \alpha}$.
Moreover, this holds independently of the events of other nodes being colored.

In total, after $\lceil \frac{2\cdot(2+\eps)}{\eps}\rceil \cdot \log d$ iterations, we get that (independently of the events of other nodes being colored), we have
$$Pr[\text{v is not colored}] \le (1-\frac{\eps \cdot \alpha}{2(2+\eps)\cdot \alpha})^{\lceil\frac{2\cdot(2+\eps)}{\eps}\rceil \log d}\le (\frac{1}{4})^{\log d} \le d^{-2}.$$
\end{proof}

At this point, we apply \Cref{sec:Coloring for Low-Arboricity Graphs:lemma:path length} to conclude that in the remaining graph there is no path longer than $O(\log n)$,
with high probability. Then, we apply the $O(\log n)$-round
deterministic algorithm of \Cref{sec:Low-Arb-Coloring:deterministic-coloring}, to color the remaining graph with $d+1$ extra colors.

\subsection{Deterministic Coloring}\label{sec:Coloring for Low-arbirtrary Graphs:deterministic step}
After we partially color the input graph $G$ with either of the algorithms of~\Cref{sec:Low-Arboriciry Partial Coloring} and~\Cref{sec: Low-Arboricity Coloring:O(a) algorithm},
in the remaining graph there is no path longer than $O(\log n)$, with high probability. 

In this section, we color deterministically the remaining graph as follows.
Each remaining (uncolored) node receives $d+1$ new colors and performs the following algorithm.

\begin{framed}\label{sec:Coloring for Low-Arboricity Graphs:deterministic algorithm}
{\setlength{\parindent}{0cm}
\textbf{Low-Arb Deterministic Coloring Algorithm, run by each uncolored node $v$: }
\begin{itemize}
 \item Node $v$ waits for all its remaining out-neighbors to be colored and removes their colors from its palette.
 \item It gets permanently colored with one remaining color $x$, and informs its neighbors.
\end{itemize}}
\end{framed}
\begin{lemma}\label[lemma]{sec:Low-Arb-Coloring:deterministic-coloring}
 After $O(\log n)$ rounds, every node is colored, with high probability.
\end{lemma}
\begin{proof}
Consider a remaining (uncolored) node $v$ that runs the above algorithm. 
Since it has at most $d$ remaining out-neighbors, there is always at least one available color to select the moment that we color node $v$.
Furthermore,
by \Cref{sec:Coloring for Low-Arboricity Graphs:lemma:path length}, 
there is no path longer than $O(\log n)$ in the remaining graph, with high probability; this implies that with high probability,
$v$ does not wait more
than $O(\log n)$ rounds until it gets permanently colored.
\end{proof}

\bibliographystyle{plainurl}
\bibliography{ref}

\begin{thebibliography}{10}

\bibitem{Awerbuch:1989:NDL}
Baruch Awerbuch, M~Luby, AV~Goldberg, and Serge~A Plotkin.
\newblock Network decomposition and locality in distributed computation.
\newblock In {\em Foundations of Computer Science, 1989., 30th Annual Symposium
  on}, pages 364--369, 1989.

\bibitem{Barenboim:2015:D9S}
Leonid Barenboim.
\newblock Deterministic (${\Delta}$+ 1)-coloring in sublinear (in ${\Delta}$)
  time in static, dynamic, and faulty networks.
\newblock {\em Journal of the ACM (JACM)}, page~47, 2016.

\bibitem{barenboim2008sublogarithmic}
Leonid Barenboim and Michael Elkin.
\newblock Sublogarithmic distributed mis algorithm for sparse graphs using
  nash-williams decomposition.
\newblock In {\em Proceedings of the Twenty-seventh ACM Symposium on Principles
  of Distributed Computing}, PODC '08, pages 25--34, 2008.

\bibitem{Barenboim:2009:DLC}
Leonid Barenboim and Michael Elkin.
\newblock Distributed (${\Delta}$+ 1)-coloring in linear (in ${\Delta}$) time.
\newblock In {\em Proceedings of the forty-first annual ACM symposium on Theory
  of computing}, pages 111--120, 2009.

\bibitem{barenboim2010sublogarithmic}
Leonid Barenboim and Michael Elkin.
\newblock Sublogarithmic distributed {M}{I}{S} algorithm for sparse graphs
  using nash-williams decomposition.
\newblock {\em Distributed Computing}, 22(5-6):363--379, 2010.

\bibitem{Barenboim:2010:DDV}
Leonid Barenboim and Michael Elkin.
\newblock Deterministic distributed vertex coloring in polylogarithmic time.
\newblock {\em Journal of the ACM (JACM)}, page~23, 2011.

\bibitem{barenboim2013monograph}
Leonid Barenboim and Michael Elkin.
\newblock Distributed graph coloring: Fundamentals and recent developments.
\newblock {\em Synthesis Lectures on Distributed Computing Theory},
  4(1):1--171, 2013.

\bibitem{barenboim2014distributed}
Leonid Barenboim, Michael Elkin, and Fabian Kuhn.
\newblock Distributed (${\Delta}$+1)-coloring in linear (in ${\Delta}$) time.
\newblock {\em SIAM Journal on Computing}, 43(1):72--95, 2014.

\bibitem{barenboim2012locality}
Leonid Barenboim, Michael Elkin, Seth Pettie, and Johannes Schneider.
\newblock The locality of distributed symmetry breaking.
\newblock In {\em Foundations of Computer Science (FOCS) 2012}, pages 321--330,
  2012.

\bibitem{bollobas1978chromatic}
B{\'e}la Bollob{\'a}s.
\newblock Chromatic number, girth and maximal degree.
\newblock {\em Discrete Mathematics}, 24(3):311--314, 1978.

\bibitem{chung2014LLL}
Kai-Min Chung, Seth Pettie, and Hsin-Hao Su.
\newblock Distributed algorithms for the lov{\'a}sz local lemma and graph
  coloring.
\newblock In {\em Proceedings of the 2014 ACM symposium on Principles of
  distributed computing}, pages 134--143, 2014.

\bibitem{ColeVishkin}
Richard Cole and Uzi Vishkin.
\newblock Deterministic coin tossing and accelerating cascades: micro and macro
  techniques for designing parallel algorithms.
\newblock In {\em Proceedings of the eighteenth annual ACM symposium on Theory
  of computing}, pages 206--219. ACM, 1986.

\bibitem{fraigniaud2016local}
Pierre Fraigniaud, Marc Heinrich, and Adrian Kosowski.
\newblock Local conflict coloring.
\newblock In {\em Foundations of Computer Science (FOCS), 2016 IEEE 57th Annual
  Symposium on}, pages 625--634. IEEE, 2016.

\bibitem{harris2016distributed}
David~G Harris, Johannes Schneider, and Hsin-Hao Su.
\newblock Distributed (${\Delta}$+ 1)-coloring in sublogarithmic rounds.
\newblock In {\em Proceedings of the 48th Annual ACM SIGACT Symposium on Theory
  of Computing}, pages 465--478, 2016.

\bibitem{johansson1999simple}
{\"O}jvind Johansson.
\newblock Simple distributed ${\Delta}$+ 1-coloring of graphs.
\newblock {\em Information Processing Letters}, 70(5):229--232, 1999.

\bibitem{Kothapalli:2011}
Kishore Kothapalli and Sriram Pemmaraju.
\newblock Distributed graph coloring in a few rounds.
\newblock In {\em Proceedings of the 30th Annual ACM SIGACT-SIGOPS Symposium on
  Principles of Distributed Computing}, PODC '11, pages 31--40, 2011.

\bibitem{Kuhn:2009:WGC}
Fabian Kuhn.
\newblock Weak graph colorings: Distributed algorithms and applications.
\newblock In {\em Proceedings of the Twenty-first Annual Symposium on
  Parallelism in Algorithms and Architectures}, pages 138--144, 2009.

\bibitem{linial1987LOCAL}
Nathan Linial.
\newblock Distributive graph algorithms global solutions from local data.
\newblock In {\em Proc.\ of the Symp.\ on Found.\ of Comp.\ Sci.\ (FOCS)},
  pages 331--335. IEEE, 1987.

\bibitem{linial1992locality}
Nathan Linial.
\newblock Locality in distributed graph algorithms.
\newblock {\em SIAM Journal on Computing}, 21(1):193--201, 1992.

\bibitem{luby1986simple}
Michael Luby.
\newblock A simple parallel algorithm for the maximal independent set problem.
\newblock {\em SIAM journal on computing}, 15(4):1036--1053, 1986.

\bibitem{nash1964decomposition}
CSJA Nash-Williams.
\newblock Decomposition of finite graphs into forests.
\newblock {\em Journal of the London Mathematical Society}, 1(1):12--12, 1964.

\bibitem{panconesirizzi2000}
Alessandro Panconesi and Romeo Rizzi.
\newblock Some simple distributed algorithms for sparse networks.
\newblock {\em Distributed computing}, 14(2):97--100, 2001.

\bibitem{panconesi1992improved}
Alessandro Panconesi and Aravind Srinivasan.
\newblock Improved distributed algorithms for coloring and network
  decomposition problems.
\newblock In {\em Proc.\ of the Symp.\ on Theory of Comp.\ (STOC)}, pages
  581--592. ACM, 1992.

\bibitem{Peleg:2000}
David Peleg.
\newblock {\em Distributed Computing: A Locality-sensitive Approach}.
\newblock Society for Industrial and Applied Mathematics, Philadelphia, PA,
  USA, 2000.

\bibitem{pettie2015distributed}
Seth Pettie and Hsin-Hao Su.
\newblock Distributed coloring algorithms for triangle-free graphs.
\newblock {\em Information and Computation}, 243:263--280, 2015.

\bibitem{Schneider:2010}
Johannes Schneider and Roger Wattenhofer.
\newblock A new technique for distributed symmetry breaking.
\newblock In {\em Proceedings of the 29th ACM SIGACT-SIGOPS symposium on
  Principles of distributed computing}, pages 257--266, 2010.

\bibitem{Schneider2011}
Johannes Schneider and Roger Wattenhofer.
\newblock Distributed coloring depending on the chromatic number or the
  neighborhood growth.
\newblock In {\em International Colloquium on Structural Information and
  Communication Complexity}, pages 246--257. Springer, 2011.

\bibitem{Szegedy:1993}
M\'{a}ri\'{o} Szegedy and Sundar Vishwanathan.
\newblock Locality based graph coloring.
\newblock In {\em Proceedings of the Twenty-fifth Annual ACM Symposium on
  Theory of Computing}, pages 201--207, 1993.

\end{thebibliography}

\appendix

\end{document}